\documentclass[a4paper,USenglish]{article}

\title{\textbf{
Detecting cliques in \textsf{CONGEST} networks}\thanks{Research partially supported by the Centre for Discrete Mathematics and its Applications (DIMAP), by EPSRC award EP/D063191/1, and by EPSRC award EP/N011163/1.
}}

\usepackage{amsthm,amsmath,amssymb,amsfonts}
\usepackage{graphicx}
\usepackage{latexsym}
    \usepackage{times} 
        \usepackage[scaled=.90]{helvet}
        \usepackage{courier}
\usepackage{fancybox}
\usepackage{fullpage}
\usepackage{color}
\usepackage{paralist}
\usepackage{environ}
\usepackage{xstring}
\usepackage{xspace}
\usepackage{amsmath}
\usepackage{amsthm}
\usepackage{authblk}

\usepackage[T1]{fontenc}\usepackage{yfonts}

\newcount\shortyear\newcount\shorthour\newcount\shortminute
\shorthour=\time\divide\shorthour by 60\shortyear=\shorthour
\multiply\shortyear by 60\shortminute=\time\advance\shortminute by
-\shortyear \shortyear=\year\advance\shortyear by -1900

\def\zeit{\number\shorthour:\ifnum\shortminute<10 0\number\shortminute
\else\number\shortminute\fi}

\let\oldproofname=\proofname
\renewcommand{\proofname}{\rm\bf{\oldproofname}}

\newtheorem{theorem}{Theorem}

\newtheorem{lemma}[theorem]{Lemma}

\newtheorem{definition}[theorem]{Definition}

\newcommand{\Exp}{\mathbb{E}}
\newcommand{\Var}{\mathbb{V}}
\renewcommand{\Pr}{\mathbb{P}}

\newcommand{\Order}{\mathrm{O}}
\renewcommand{\Order}{\ensuremath{\mathcal{O}}}
\newcommand{\OrderT}{\widetilde{\Order}}
\newcommand{\ThetaT}{\widetilde{\mathrm{\Theta}}}
\newcommand{\OmegaT}{\widetilde{\mathrm{\Omega}}}

\usepackage{calc}
\usepackage{mathtools}

\def\epsilon{\ensuremath{\varepsilon}}

\newcommand{\COMMENTED}[1]{{}}
\newcommand{\junk}[1]{\COMMENTED{#1}}

\newlength{\savedparindent}
\newcommand{\SaveIndent}{\setlength{\savedparindent}{\parindent}}
\newcommand{\RestoreIndent}{\setlength{\parindent}{\savedparindent}}

\newcommand{\InGray}[1]{%
\SaveIndent{} %
\noindent{} \fcolorbox[rgb]{0,0,0}{0.95,0.95,0.95}{
\begin{minipage}{0.965\linewidth} %
\RestoreIndent{}%
#1
\end{minipage}
} }

\newcommand{\InGrayMiddle}[1]{%
\SaveIndent{} %
\centerline{ \fcolorbox[rgb]{0,0,0}{0.95,0.95,0.95}{
\begin{minipage}{0.8\linewidth} %
\RestoreIndent{}
#1
\end{minipage}
} } }

\NewEnviron{algo}{\medskip\InGrayMiddle{\mbox{}\\\BODY\mbox{}\\[-0.35in]}\medskip}
\NewEnviron{walgo}{\begin{figure}[t]\InGray{\mbox{}\\\BODY\mbox{}\\[-0.3in]}\end{figure}}

\newcommand{\congest}{\textsf{CONGEST}\xspace}
\newcommand{\congestb}{\textsf{CONGEST}$_{\B}$\xspace}
\newcommand{\congc}{\textsf{CONGESTED CLIQUE}\xspace}
\newcommand{\local}{\textsf{LOCAL}\xspace}

\newcommand{\B}{\ensuremath{\mathfrak{b}}}
\newcommand{\EPS}{\ensuremath{\mathcal{E}}}

\newcommand{\DISJ}{\ensuremath{\textsc{DISJ}_n}}

\newcommand{\cut}{\ensuremath{\mathcal{C}}}
\newcommand{\scut}{\ensuremath{|\cut|}}

\title{Detecting cliques in \textsf{CONGEST} networks\thanks{A.C. is supported by the Centre for Discrete Mathematics and its Applications (DIMAP), by EPSRC award EP/D063191/1, and by EPSRC award EP/N011163/1. C.K. carried out most of the work on this paper while he was at the University of Warwick, where he was supported by the Centre for Discrete Mathematics and its Applications (DIMAP), and by EPSRC award EP/N011163/1.}}

\author[1]{Artur Czumaj}

\author[2]{Christian Konrad}

\affil[1]{Department of Computer Science and Centre for Discrete Mathematics and its Applications (DIMAP), University of Warwick, UK, \texttt{A.Czumaj@warwick.ac.uk}}
\affil[2]{Department of Computer Science, University of Bristol, UK, \texttt{christian.konrad@bristol.ac.uk}}

\begin{document}


\maketitle

\begin{abstract}
The problem of detecting network structures plays a central role in distributed computing. One of the fundamental problems studied in this area is to determine whether for a given 
graph $H$, the input network contains a subgraph isomorphic to $H$ or not. We investigate this problem for $H$ being a clique $K_{\ell}$ in the classical distributed \congest model, where the communication topology is the same as the topology of the underlying network, and with limited communication bandwidth on the links.


Our first and main result is a lower bound, showing that detecting $K_{\ell}$ requires $\Omega(\sqrt{n} / \B)$ communication rounds, for every $4 \le \ell \le \sqrt{n}$, and $\Omega(n / (\ell \B))$ rounds for every $\ell \ge \sqrt{n}$, where $\B$ is the bandwidth of the communication links. This result is obtained by using a reduction to the set disjointness problem in the framework of two-party communication complexity.
We complement our lower bound with a two-party communication protocol for listing all cliques in the input graph, which up to constant factors communicates the same number of bits as our lower bound for $K_4$ detection. This demonstrates that our lower bound cannot be improved using the two-party communication framework.
\end{abstract}


\section{Introduction}

We study the problem of detecting network structures in a distributed environment, which is a fundamental problem in modern computing. Our focus is on the \emph{subgraph detection problem}, in which for a given 
graph $H$, one wants to determine whether the network graph $G$ contains a subgraph isomorphic to $H$ or not. We investigate this problem for $H$ being a clique $K_{\ell}$ for $\ell \ge 4$.


The nowadays classical distributed \emph{\congest model} (see, e.g., \cite{P00}) is a variant of the classical \local model of distributed computation (where in each round network nodes can send through all incident links messages of unrestricted size) with limited communication bandwidth. The distributed system is represented as a network (undirected graph) $G = (V,E)$ with $n = |V|$ nodes, where each node $v \in V$ executes the same algorithm in synchronous rounds, and the nodes collaborate to solve a graph problem with input $G$. Each node is assumed to have a unique identifier from $\{0, \dots, \text{poly}(n)\}$. In any single round, all nodes can:
\begin{enumerate}[(i)]
\item perform an unlimited amount of local computation,
\item send a possibly different $\B$-bit message to each of their neighbors, and
\item receive all messages sent to them.
\end{enumerate}
We measure the \emph{complexity} of an algorithms by the number of synchronous rounds required.

In accordance with the standard terminology in the literature, we assume $\B = \Order(\log n)$; we note though that our analysis generalizes to other settings of $\B$ in a straightforward manner. (We note that in our lower bound for detecting $K_4$ and $K_{\ell}$ in Section \ref{sec:lower-bound-for-K4}, to ensure full generality of presentation, we will make the analysis parametrized by the message size $\B$, in which case we will refer to such model of distributed computation as \congestb, the \congest model with messages of size $\B$.)

Our goal is, for a given network $G = (V,E)$ and $\ell \ge 4$, to solve the \emph{subgraph detection problem} for a clique $K_{\ell}$, that is, to design an algorithm in the \congest model such that
\begin{enumerate}[(i)]
\item if $G$ contains a copy of $K_{\ell}$, then with probability at least $\frac23$ at least one node outputs 1, and
\item if $G$ does not contain any copy of $K_{\ell}$, then with probability at least $\frac23$ no node outputs 1.
\end{enumerate}

The subgraph detection problem is a local problem: it can be solved efficiently solely on the basis of local information. In particular, in the \congest model, the problem of finding $K_{\ell}$ in a graph can be trivially solved in $\Order(n)$ rounds, or in fact, in $\Order(\max_{u \in V} \deg_G(u))$ rounds, where $\deg_G(u)$ denotes the degree of node $u$ in $G$. Indeed, if each node sends its entire neighborhood to all its neighbors, then afterwards, each node will be aware of all its neighbors and of their neighbors. Therefore, in particular, each node will be able to detect all cliques it belongs to. Since for each node $u$, the task of sending its entire neighborhood to all its neighbors can be performed in $\Order(\deg_G(u))$ rounds in the \congest model, the total number of rounds for the entire network is $\Order(\max_{u \in V} \deg_G(u)) = \Order(n)$ rounds. In view of this simple observation, the main challenge in the clique $K_{\ell}$ detection problem 
is whether this task can be performed in a \emph{sublinear number of rounds.}


\subsection{Our results}

In this paper, we give the first non-trivial lower bound for the complexity of detecting a clique $K_{\ell}$ in the \congestb model, for $\ell \ge 4$. In Theorem \ref{thm:lb-Kl}, we prove that every algorithm in the \congestb model that with probability at least $\frac23$ detects $K_{\ell}$, for $\ell \ge 4$ and $\ell = \Order(\sqrt{n})$, requires $\Omega(\sqrt{n}/\B)$ rounds. Further, if $\ell = \omega(\sqrt{n})$, then $\Omega(n/(\ell\,\B))$ rounds are required. We are not aware of any other non-trivial (super-constant) lower bound for this problem in the \congestb model.

We complement our lower bound with a two-party communication protocol for listing all cliques in the input graph (see Theorem \ref{thm:upper-bound}), which up to constant factors communicates the same number of bits as our lower bound for $K_4$ detection. This demonstrates that our lower bound is essentially tight in this framework, and cannot be improved using the two-party communication approach.


\subsection{Techniques: Framework of two-party communication complexity}
\label{subsec:techniques}

Our main results, the lower bound of clique detection in Theorem \ref{thm:lb-Kl} and the upper bound in Theorem \ref{thm:upper-bound}, rely on the \emph{two-party communication complexity} framework and the use of a tight lower bound for the set disjointness problem in this framework.

We consider the classical two-party communication complexity setting (cf. \cite{KN97}) in which two players, Alice and Bob, each have some private input $X$ and $Y$. The players' goal is to compute a joint function $\mathfrak{f}(X,Y)$, and the complexity measure used is the number of bits Alice and Bob must exchange to compute $\mathfrak{f}(X,Y)$. In the two-party communication problem of \emph{set disjointness}, Alice's input is $X \in \{0, 1\}^n$ and Bob holds $Y \in \{0, 1 \}^n$, and their goal is to compute $\DISJ(X,Y) := \bigvee_{i=1}^n X_i \wedge Y_i$. In a seminal work, Kalyanasundaram and Schnitger \cite{KS92} showed that in any randomized communication protocol, the players must exchange $\Omega(n)$ bits to solve the set disjointness problem with constant success probability.

\begin{theorem}[\cite{KS92}]
\label{thm:comm}
The randomized two-party communication complexity of set disjointness 
is $\Omega(n)$. That is, for any constant $p>0$, any randomized two-party communication protocol that computes $\DISJ(X,Y)$ with probability at least $p$, has two-party communication complexity $\Omega(n)$.
\end{theorem}

Our main result, the lower bound for detecting $K_{\ell}$ in the \congest model, relies on a reduction from the two-party communication problem set disjointness. The two-party communication framework, and, in particular, the two-party set disjointness problem, have been frequently used in the past to construct lower bounds for the \congest model, see, e.g., \cite{CKP17,DKO14,FGO17b,GO17,KR17}. A typical approach relies on a construction of a special graph $G = (V,E)$ with some fixed edges and some edges depending on the input of Alice and Bob. One partitions the nodes of $G$ into two disjoint sets $V_A$ and $V_B$. Let $\cut$ be the $(V_A, V_B)$-cut, that is, the set of edges in $G$ with one endpoint in $V_A$ and one endpoint in $V_B$. Let $E_A$ be the edge set of $G[V_A]$ (subset of $E$ on vertex set $V_A$) and $E_B$ be the edge set of $G[V_B]$. We consider a scenario where Alice's input is represented by the subgraph $G_A=(V, E_A \cup \cut) \subseteq G$ and Bob's input is represented by $G_B = (V, E_B \cup \cut) \subseteq G$. (We denote this way of distributing the vertex and edge sets as the \emph{vertex partition model}.) In order to learn any information about the structure of $G[A] \setminus \cut$ and $G[B] \setminus \cut$, and hence about the input of the other player, Alice and Bob must communicate through the edges of the cut $\cut$. Therefore, in order to obtain a lower bound for a problem in the \congestb model, one wants to construct $G$ to ensure that it has some property (in our case, contains a copy of $K_{\ell}$) if and only if the corresponding instance of set disjointness is such that $\DISJ(X,Y) = 1$, and in order to determine the required property, one has to communicate a large part of (essentially the entire graph) $G[A]$ 
through $\cut$. With this approach, if the cut $\cut$ has size $\scut$, and the private inputs of Alice and Bob (edges in $G[A] \setminus \cut$ or $G[B] \setminus \cut$) are of size $\mathfrak{s}$, one can apply Theorem \ref{thm:comm} to argue that the round complexity of any distributed algorithm in the \congestb model for a given problem is $\Omega(\frac{\mathfrak{s}} {\scut \cdot \B})$. The central challenge is to ensure that for the encoded set disjointness instance of size $\mathfrak{s}$ and the cut of size $\scut$, the ratio $\frac{\mathfrak{s}}{\scut}$ is as large as possible.


For example, Drucker et al. \cite{DKO14} incorporated a similar approach to obtain a lower bound for the subgraph detection problem in a \emph{broadcast} variant of the \congestb model (in fact, even for a (stronger) broadcast variant of the \congc model), where nodes are required to send the same message through all their incident edges. The lower bound construction requires sending $\Omega(n^2)$ bits through the cut of size $\Order(n^2)$, but the fact that in the broadcast variant of the \congestb model every node is required to send the same message via all incident edges, at most $\Order(n \, \B)$ bits can be transmitted through the cut, yielding a lower bound of $\Omega(\frac{n}{\B})$. (In particular, for the broadcast variant of the \congestb model, Drucker et al.\ \cite[Theorem~15]{DKO14} proved that detecting a clique $K_{\ell}$, $\ell \ge 4$, requires $\Omega(\frac{n}{\B})$ rounds.) Note however that in the (non-broadcast) \congestb model, this construction does not give any not-trivial bound, since $\frac{\mathfrak{s}}{\scut} = \Order(1)$.

Our main building block for our lower bound is the construction of $(\Omega(n^2), \Order(n^{3/2}))$-lower-bound graphs (in Section \ref{subsec:construction-lower-bound-graphs}) that can be used to encode a set disjointness instance of size $\mathfrak{s} = \Omega(n^2)$ such that the cut is of size $\scut = \Order(n^{3/2})$. By incorporating these bounds in the framework described above, this construction leads to the first non-trivial lower bound of $\Omega(\frac{\sqrt{n}}{\B})$ for the subgraph detection problem in the \congestb model for the clique $K_4$. This construction can also be extended to detect larger cliques, yielding the lower bound of $\Omega(\frac{n}{(\ell + \sqrt{n}) \, \B})$ for detecting any $K_{\ell}$ with $\ell \ge 4$.

Since these are the first superconstant lower bounds for detecting a clique in the \congest model and since the best upper bound for these problems is still $\Order(n)$, the next goal is to understand to what extent these bounds could be improved and whether the existing approach could be used for that task. Do we need $\Omega(\frac{\sqrt{n}}{\B})$ communication rounds to detect any clique $K_{\ell}$ (with $\ell \ge 4$, $\ell = \Order(\sqrt{n})$) in the \congestb model, or maybe we need as many as a linear number of rounds? While we do not know the answer to this question, and in fact, this question is the main open problem left by this paper, we can prove that any better lower bound would require a significantly different approach, going beyond the two-party communication framework in the vertex partition model.

Indeed, let us consider the vertex partition model in the two-party communication framework, as defined above. The input consists of an undirected $G=(V, E)$ with an arbitrary vertex partition $V = V_A \ \dot{\cup} \ V_B$. We consider a scenario where Alice is given the subgraph $G_A=(V, E_A \cup \cut) \subseteq G$ and Bob is given $G_B = (V, E_B \cup \cut) \subseteq G$, where $\cut$ is the $(V_A, V_B)$-cut in $G$.
The arguments in our construction of lower-bound graphs in Theorem~\ref{thm:lower-bound-graph} imply that for some inputs, any two-party communication protocol in the vertex partition model for the problem of listing all cliques in a given graph with $n$ nodes requires communication of $\Omega(\sqrt{n} \, \scut)$ bits between Alice and Bob. We will prove in Section~\ref{sec:upper-bound} (Theorem~\ref{thm:upper-bound}) that this lower bound is asymptotically tight in the two-party communication framework in the vertex partition model. We show that there is a two-party communication protocol in the vertex partition model for listing \emph{all cliques}
that uses $\Order(\sqrt{n} \, \scut)$ communication rounds, where $\cut$ is the set of shared edges between Alice and Bob. This shows that we cannot obtain stronger lower bounds for the $K_{\ell}$-detection problem, for $\ell = \Order(\sqrt{n})$, in the \congest model using the two-party communication framework in the vertex partition model.


\subsection{Related works}


%
As a fundamental primitive, subgraph detection and listing in the \congest model has been recently receiving attention from multiple authors, focusing mainly on randomized complexity. However, despite major efforts, for the \congest model, relatively little is known about the complexity of the subgraph detection problem.

Rather surprisingly, prior to our work, no non-trivial results about the complexity of clique $K_{\ell}$ ($\ell \ge 4$) detection in the \congest model have been known. While there is a trivial lower bound of a constant number of rounds, and as we mentioned earlier, one can easily solve the problem in $\Order(n)$ rounds in the \congest model, no sublinear upper bounds nor superconstant lower bounds have been known.

In a recent breakthrough in this area, Izumi and Le Gall \cite{IL17} raised some hopes that maybe these problems could be solved in a sublinear number of rounds in the \congest model. They considered the subgraph detection problem for the smallest interesting subgraph $H$, the triangle $K_3$, and presented a very clever algorithm that detects a triangle in $\OrderT(n^{2/3})$ rounds. Further, Izumi and Le Gall \cite{IL17} also showed that the related problem of finding all triangles (triangle listing) can be solved in $\OrderT(n^{3/4})$ rounds. There is no non-trivial lower bound for the triangle detection problem, though it is known (cf. \cite{IL17,POS16}) that the more complex triangle listing problem requires $\Omega(n^{1/3}/\log n)$ rounds, even in the \congc model. It can also be shown that the problem of listing all triangles such that each node $v$ learns all triangles that it is part of significantly harder than the general triangle listing problem and requires $\Omega(n / \log n)$ rounds \cite[Proposition~4.4]{IL17}.
While rather disappointingly, we do not know how to extend any of these upper bounds to other cliques $K_{\ell}$ with $\ell \ge 4$, the work of Izumi and Le Gall \cite{IL17} raises hope that detecting cliques $K_{\ell}$ could potentially be solved in a sublinear number of rounds. In fact, even for $K_3$, we do not even know whether detecting a triangle $K_3$ can be solved in a polylogarithmic or even a constant number of rounds in the \congest model (the lower bound of $\Omega(n^{1/3}/\log n)$ rounds in the \congc model (cf. \cite{IL17,POS16}) holds only for a more complex problem of detecting \emph{all triangles}).

Even et al.\ \cite{EFFGLMMOORT17} noted that the problem is significantly simpler for trees, and designed a randomized color-coding algorithm that detects any constant-size \emph{tree} on $\ell$ nodes 
in $\Order(\ell^{\ell})$ rounds.

As for lower bounds for the subgraph detection problem in the \congest model, until very recently, the only hardness results known in the literature have been for cycles. For any fixed $ \ge 4$, there is a polynomial lower bound for detecting the $\ell$-cycle $C_{\ell}$ in the \congest model \cite{DKO14}, where it has been shown that detecting $C_{\ell}$ requires $(\text{ex}(n,C_{\ell})/ \log n)$ rounds, where $\text{ex}(n,C_{\ell})$ is the Tur\'{a}n number for cycles, that is, the largest possible number of edges in a $C_{\ell}$-free graph over $n$ vertices. In particular, for odd-length cycles (of length 5 or more), the lower bound of \cite{DKO14} is $\Omega(n/\log n)$, and it is $\Omega(\sqrt{n} / \log n)$ for $\ell = 4$. Very recently, Korhonen and Rybicki \cite{KR17} improved the lower bound for all even-length cycles to $\Omega(\sqrt{n} / \log n)$.
Further, Gonen and Oshman \cite{GO17} extended these lower bounds for $C_{\ell}$-freeness to some related classes of graphs, though still with some cyclic underlying structure. (As mentioned above, we note that Drucker et al.\ \cite{DKO14} presented lower bounds for other graphs, but this was in a \emph{broadcast} variant of the \congc model, where nodes are required to send the same message on all their edges. In particular, for the broadcast variant of the \congc model, Drucker et al.\ \cite{DKO14} proved that detecting a clique $K_{\ell}$, $\ell \ge 4$, requires $\Omega(n / \log n)$ rounds.)

The only lower bound for the subgraph detection problem for $H$ significantly other than cycles, is a very recent work of Fischer et al.\ \cite{FGO17b}, who demonstrated that the subgraph detection problem is hard even for some subgraphs $H$ of constant size. In particular, for any constant $\ell \ge 2$, there is a graph $H$ with a constant number of vertices and edges such that the problem of finding $H$ in a network of size $n$ requires time $\Omega(n^{2-\frac{1}{\ell}}/\B)$ in the \congest model, where $\B$ is the bandwidth of each communication links.

%
There has also been some recent research for the \emph{deterministic} subgraph detection problem in the \congest model. For example, Drucker et al.\ \cite{DKO14} designed an $\Order(\sqrt{n})$ round algorithm for $C_4$ detection, and Even et al.\ \cite{EFFGLMMOORT17} and Korhonen and Rybicki \cite{KR17} obtained path and tree detection algorithms requiring only a constant number of rounds. Korhonen and Rybicki \cite{KR17} considered also deterministic subgraph detection (for paths, cycles, trees, pseudotrees, and on $d$-degenerate graphs) in the weaker broadcast \congest model, where nodes send the same message to all neighbors in each communication round. In the \congc model, deterministic subgraph detection algorithms were given by Dolev et al.\ \cite{DLP12} and Censor-Hillel et al.\ \cite{CKKLPS15}.


We summarize earlier results together with our new results in Table~\ref{table}.

{\renewcommand{\arraystretch}{1.2}
\begin{table}[t]
\centering
\begin{tabular}{||c|c|c|c||}\hline\hline
Paper & Time bound & Problem & Model
            \\ \hline
\cite{EFFGLMMOORT17} & $\Order(\ell^{\ell})$ & Detecting a tree on $\ell$ nodes & \congest
            \\
folklore & $\Order(n)$ & Detecting $K_{\ell}$, $\ell \ge 3$ & \congest
            \\
\cite{IL17} & $\OrderT(n^{2/3})$ & Detecting triangle $K_3$ & \congest
            \\
\cite{IL17} & $\OrderT(n^{3/4})$ & Triangle listing & \congest
            \\\hline
\cite{FGO17b} & $\Omega(n^{2-\frac1{\ell}}/\log n)$ & Detecting some $H$ of size $\Order(\ell)$ & \congest
            \\
\cite{DKO14} & $\Omega(n/\log n)$ & Detecting $C_{\ell}$, $\ell \ge 5$, $\ell$ odd & \congest
            \\
\cite{DKO14,KR17} & $\Omega(\sqrt{n}/\log n)$ & Detecting $C_{\ell}$, $\ell \ge 4$, $\ell$ even & \congest
            \\
\cite{IL17,POS16} & $\Omega(n^{1/3}/\text{poly-log}(n))$ & Triangle listing & \congc
            \\
\cite{DKO14} & $\Omega(n / \log n)$ & Detecting $K_{\ell}$ for $\ell \ge 4$ & broadcast \congc
            \\
Theorem \ref{thm:lb-K4} & $\Omega(\sqrt{n}/\log n)$ & Detecting $K_4$ & \congest
			\\
Theorem \ref{thm:lb-Kl} & $\Omega(\sqrt{n}/(\ell \log n))$ & Detecting $K_{\ell}$ for $\ell \ge 4$ & \congest
			\\
            \hline\hline
\end{tabular}
\caption{Prior (randomized) results for the problem of detecting a given subgraph $H$, or for listing all copies of $H$, in the \congest model (less relevant results (upper bounds) for the \congc model are omitted; note that lower bounds for \congc hold also for \congest and lower bounds for broadcast \congc do not imply any bounds for \congest).}
\label{table}
\end{table}
}


\subsubsection{Property testing of $H$-freeness}
Since there have been so few positive results for the original subgraph detection problem, recently there have been some advances in a relaxation of this problem, a closely related (and significantly simpler) problem of \emph{testing subgraphs freeness} in the \emph{framework of property testing for distributed computations} (see, e.g., \cite{BP11,EFFGLMMOORT17}).
In the property testing setting, an algorithm has to decide, with probability at least $\frac23$, if the input graph is (a) $H$-free (i.e., does not contain a subgraph isomorphic to $H$) or (b) $\varepsilon$-far from being $H$-free (that is, the goal is to distinguish whether the input graph $G$ is $H$-free or one needs to modify more than $\varepsilon |E(G)|$ edges of $G$ to obtain a graph that is $H$-free); in the intermediate case, the algorithm can perform arbitrarily (see e.g., \cite{CKKLPS15,EFFGLMMOORT17} for more details). Property testing of $H$-freeness in the \congest model has received a lot of attention lately (see, e.g., \cite{BP11,CFSV16,EFFGLMMOORT17,FGO17b,FO17}). In particular, it has been shown \cite{EFFGLMMOORT17} that testing $H$-freeness can be done in $\Order(1/\varepsilon)$ round in the \congest model for any constant-size graph $H$ containing an edge $(x,y)$ such that any cycle in $H$ contains at least one of $x, y$. This implies testing in $\Order(1/\varepsilon)$ rounds of any cycle $C_k$, and of any subgraph $H$ on five (or less) vertices except $K_5$. Further, for any $\ell \ge 5$, $K_{\ell}$-freeness can be tested in 
$\Order((\varepsilon \cdot |E(G)|)^{\frac12 - \frac{1}{\ell-2}}/\varepsilon)$ rounds \cite{EFFGLMMOORT17}. For trees, Even et al.\ \cite{EFFGLMMOORT17} show that testing if the input graph is $T$-free for a tree $T$ on $\ell$ vertices can be done in $\Order(\ell^{1+\ell^2}/\varepsilon^{\ell})$ rounds the \congest model.


\junk{
\subsection{Outline}

Our lower bound is presented in Sections \ref{sec:lower-bound-for-K4}--\ref{sec:lb-graph-construction}. We begin in Section \ref{subsec:lower-bound-graphs} with a definition of lower-bound graphs and then, in Sections \ref{subsec:hardness-clique-detection}--\ref{subsec:lb-K_ell} we show how to combine lower-bound graphs and the lower bound for set disjointness to prove the hardness of clique detection.
%
%
%
%
The construction of $(\Omega(n^2), \Order(n^{3/2}))$-lower-bound graphs is presented in Section \ref{sec:lb-graph-construction}.

Our upper bound, a two-party communication protocol in the vertex partition model for listing all cliques, is presented in Section \ref{sec:upper-bound}.

Section \ref{sec:conclusions} gives some final conclusions.
}


\section{Lower bound results (detecting a clique requires $\OmegaT(\sqrt{n})$ rounds)}
\label{sec:lower-bound-for-K4}

In this section we prove our hardness results showing that any algorithm in the \congestb model that detects a $K_{\ell}$ with probability at least $\frac23$ requires $\Omega(\sqrt{n}/\B)$ rounds, for every $\ell = \Order(\sqrt{n})$ and $\ell \ge 4$, and requires $\Omega(\frac{n}{\ell \B})$ rounds if $\ell = \omega(\sqrt{n})$ (Theorems \ref{thm:lb-K4} and \ref{thm:lb-Kl}); or in short, $\Omega(\frac{n}{(\ell + \sqrt{n}) \, \B})$ rounds, for every $\ell \ge 4$. Our lower bound for the complexity of detecting $K_{\ell}$ in the \congest model relies on a reduction to the two-party communication complexity lower bound for the set disjointness problem (cf. Theorem \ref{thm:comm} in Section \ref{subsec:techniques}), which we implement with the help of lower-bound graphs (cf. Section \ref{subsec:lower-bound-graphs}).


\junk{
\subsection{Two-party communication complexity of set disjointness}
\label{subsec:two-party}

We consider the classical two-party communication setting (cf. \cite{KN97}) in which two players, Alice and Bob, each have some private input $X$ and $Y$. The players' goal is to compute a joint function $f(X,Y)$, and the complexity measure used is the number of bits Alice and Bob must exchange to compute $f(X,Y)$. In the two-party communication problem of \emph{set disjointness}, Alice's input is $X \in \{0, 1\}^n$ and Bob holds $Y \in \{0, 1 \}^n$, and their goal is to compute $\DISJ(X,Y) := \bigvee_{i=1}^n X_i \wedge Y_i$. In the seminal work, Kalyanasundaram and Schnitger \cite{KS92} showed that in any randomized communication protocol, the players must exchange $\Omega(n)$ bits to solve the set disjointness problem with constant success probability.

\begin{theorem}[\cite{KS92}]
The randomized two-party communication complexity of set disjointness 
is $\Omega(n)$. That is, for any constant $p>0$, any randomized two-party communication protocol that computes $\DISJ(X,Y)$ with probability at least $p$, has two-party communication complexity $\Omega(n)$.
\end{theorem}
}


\subsection{Lower-bound graphs}
\label{subsec:lower-bound-graphs}

Our reduction to the two-party communication complexity lower bound for the set disjointness problem relies on a notion of a \emph{lower-bound graph} (cf. Figure \ref{fig:lb-graphs}).

\begin{definition}
\label{def:lb-graph}
Let $G = (A, B, E)$ be a bipartite graph with $|A| = |B| = n$ and let $k, m$ be integers. Then $G$ is called a \emph{$(k,m)$-lower-bound graph} if:
\begin{enumerate}
\item $|E| \le m$.

\item The edge set $E$ is the union of (not necessarily disjoint) sets $\EPS_1, \EPS_2, \dots, \EPS_k$ such that, for every $i$, $1 \le i \le k$, the edge-induced subgraph $G[\EPS_i]$ \emph{is isomorphic to $K_{2,2}$}.

\item For every $i,j$, $1 \le i, j \le k$, $i \ne j$, the vertex-induced subgraph $G[A(\EPS_i) \cup B(\EPS_j)]$ is \emph{not} isomorphic to $K_{2,2}$.

\item Define two \emph{graphs associated with $G$}, $H_A = (A, E_A)$ and $H_B = (B, E_B)$. $H_A$ is the graph on vertex set $A$, where $a_1, a_2 \in A$ are adjacent if and only if there exists an index $i$ with $A(\EPS_i) = \{a_1, a_2 \}$. Similarly, $H_B$ is the graph on vertex set $B$, where $b_1, b_2 \in B$ are adjacent if and only if there exists an index $j$ with $B(\EPS_j) = \{b_1, b_2 \}$. Then, we require that $H_A$ and $H_B$ are \emph{bipartite}.

\end{enumerate}
\end{definition}


\begin{figure}[t]
\centerline
{
\includegraphics[width=0.99\textwidth]{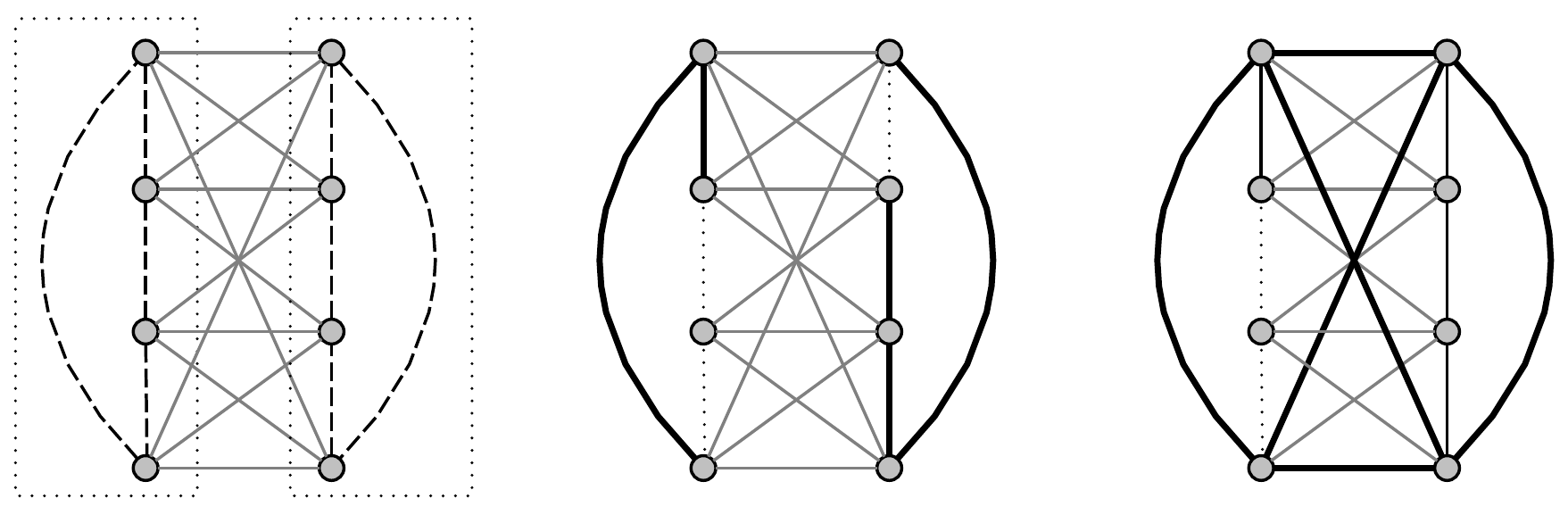}
}
\vspace{-5.7cm}

\hspace{0.9cm} $H_A$ \hspace{0.7cm} $G$ \hspace{0.7cm} $H_B$ \hspace{3.65cm} $G'$ \hspace{5.1cm} $G'$

\vspace{0.25cm}

\hspace{0.85cm} $a_1$ \hspace{2.3cm} $b_1$

\vspace{0.95cm}

\hspace{0.85cm} $a_2$ \hspace{2.3cm} $b_2$

\vspace{1.0cm}

\hspace{0.85cm} $a_3$ \hspace{2.3cm} $b_3$

\vspace{0.95cm}

\hspace{0.85cm} $a_4$ \hspace{2.3cm} $b_4$

\vspace{-3.75cm}

\hspace{6.6cm} $x_1$ \hspace{2.0cm} $y_1$

\vspace{0.85cm}

\hspace{5.55cm} $x_4$ \hspace{0.45cm} $x_2$ \hspace{2.05cm} $y_2$ \hspace{0.49cm} $y_4$

\vspace{0.85cm}

\hspace{6.6cm} $x_3$ \hspace{2.0cm} $y_3$

\vspace{1.1cm}

\caption{\textbf{Left:} Example of a $(4,12)$-lower-bound graph $G = (A, B, E)$. The dotted edges are the edges of the associated graphs $H_A$ and $H_B$ (observe that $H_A$ and $H_B$ form cycles of lengths $4$, which are bipartite). For $1 \le i \le 4$, let $\EPS_i$ be the edge set of subgraph $G[\{a_i, a_{(i \text{ mod } 4) + 1}, b_i, b_{(i \text{ mod } 4) + 1} \}]$. Observe that $E = \bigcup_{i \le 4} \EPS_i$, and, for every $i$, $G[\EPS_i]$ is isomorphic to $K_{2,2}$. Observe further that for $i \ne j$, $G[A(\EPS_i) \cup B(\EPS_j)]$ is not isomorphic to $K_{2,2}$.
\textbf{Center:} Graph $G'$ as in the proof of Theorem~\ref{thm:lb-two-party} obtained from the set disjointness instance with $X=(1,0,0,1)$ and $Y=(0,1,1,1)$. Graph $G'$ contains a $K_4$ if and only if the set disjointness instance evaluates to 1.
\textbf{Right:} The highlighted edges form a $K_4$.
}
\label{fig:lb-graphs}
\end{figure}


\subsection{Using lower-bound graphs and set disjointness to prove the hardness of clique detection}
\label{subsec:hardness-clique-detection}


With the notion of lower-bound graphs at hand, we can formalize our reduction to the two-party communication complexity lower bound for set disjointness to obtain the following central theorem.

\begin{theorem}
\label{thm:lb-two-party}
Let $G$ be a $(k,m)$-lower-bound graph. Then, detecting a $K_4$ in the \congestb model with probability at least $\frac23$ requires $\Omega\left(\frac{k}{m \B}\right)$ rounds.
\end{theorem}

\begin{proof}
Let $\mathcal{A}$ be an algorithm in the \congestb model for $K_4$ detection, that is, such that with probability at least $\frac23$, if $G$ contains a $K_4$ then at least one node outputs 1 and if $G$ contains no copy of $K_4$ then no node outputs 1. We will show that $\mathcal{A}$ can be used to solve the two-party set disjointness problem for instances of size $k$.

Consider a set disjointness instance $(X, Y)$ of size $k$. Let $G=(A, B, E)$ be a $(k,m)$-lower-bound graph, let $\EPS_1, \EPS_2, \dots, \EPS_k$ be the edge partition as in Item~2 of Definition~\ref{def:lb-graph}, and let $H_A = (A, E_A)$ and $H_B=(B, E_B)$ be the graphs associated with $G$ (Item~4 in Definition~\ref{def:lb-graph}). Alice constructs the set $E'_A \subseteq E_A$ such that for every $i$ with $X_i = 1$, the edge between $A(\EPS_i)$ is included in $E_A'$. Similarly, Bob constructs the set $E'_B \subseteq E_B$ such that for every $i$ with $Y_i = 1$, the edge between $B(\EPS_i)$ is included in $E_B'$.

We first show that the graph $G' := G \cup (E'_A \cup E'_B)$ contains a $K_4$ if and only if $\DISJ(X, Y) = 1$. 
Indeed, since by Item~4 of Definition~\ref{def:lb-graph}, the graphs $H_A$ and $H_B$ are bipartite (and thus the subgraphs $G'[A]$ and $G'[B]$ are bipartite too), any copy of $K_4$ in $G'$ must consist of two vertices from $A$ and two vertices from $B$. Let $a_1, a_2$ be any pair of distinct vertices in $A$ and $b_1, b_2$ be any pair of distinct vertices in $B$. Observe that if there is no $\EPS_i$ such that $\{a_1, a_2\} = A(\EPS_i)$ or there is no $\EPS_i$ such that $\{b_1, b_2\} = B(\EPS_i)$ then it is impossible for the nodes $a_1, a_2, b_1, b_2$ to form a $K_4$, since this would imply that either $a_1a_2 \notin E'_A$ or $b_1b_2 \notin E'_B$. Assume therefore that $\{a_1, a_2\} = A(\EPS_i)$ and $\{b_1, b_2\} = B(\EPS_j)$, for some $i,j$. Next, suppose that $i \ne j$. Then $G[\{a_1, a_2, b_1, b_2 \}]$ is not isomorphic to $K_{2,2}$, by Item~3 of Definition~\ref{def:lb-graph}. Hence, assume that $i=j$. Then $G[\{a_1, a_2, b_1, b_2 \}]$ forms a $K_{2,2}$ if and only if $X_i = Y_i = 1$, which in turn implies $\DISJ(X, Y) = 1$.

The simulation of $\mathcal{A}$ on $G'$ is executed as follows. Suppose that $\mathcal{A}$ runs in $r$ rounds. Alice simulates vertices $A$ and Bob simulates vertices $B$. In round $i$, Alice sends all messages from $A$ with destinations in $B$ to Bob, and Bob sends all messages from $B$ with destinations in $A$ to Alice. Since the cut between $A$ and $B$ is of size $m$, Alice and Bob exchange messages with overall $m \B$ bits per round. Thus, overall they communicate $r m \B$ bits. Since the algorithm allows them to solve set disjointness, by Theorem~\ref{thm:comm}, we have $rm\B = \Omega(k)$. Thus, $\mathcal{A}$ requires $\Omega(\frac{k}{m\B})$ rounds.
\end{proof}

In Theorem \ref{thm:lower-bound-graph} in Section \ref{sec:lb-graph-construction}, we prove the existence of a $(\Omega(n^2), \Order(n^{3/2}))$-lower-bound graph. By combining Theorem \ref{thm:lower-bound-graph} with Theorem \ref{thm:lb-two-party}, we obtain the following main result.

\begin{theorem}
\label{thm:lb-K4}
Every algorithm in the \congestb model that detects a $K_4$ with probability at least $\frac23$ requires $\Omega(\sqrt{n}/\B)$ rounds.
\end{theorem}


\subsection{Detection of $K_{\ell}$ for $\ell \ge 5$}
\label{subsec:lb-K_ell}


\begin{figure}[t]
\centerline
{
\includegraphics[width=0.22\textwidth]{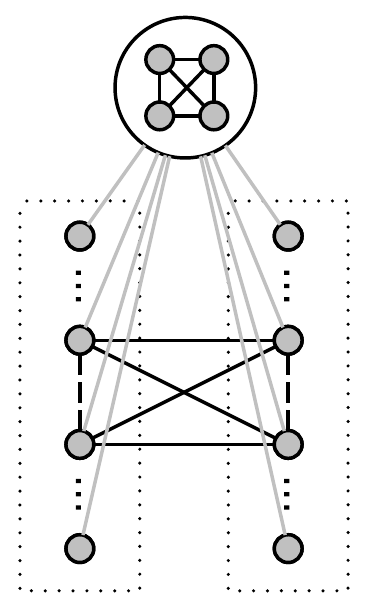}
}
\vspace{-5.6cm}

\hspace{8.75cm} $K_{\ell-4}$

\vspace{2.8cm}
\hspace{6.5cm} $x_i$ \hspace{2.0cm} $y_i$

\vspace{-2.55cm}
\hspace{6.6cm} $H_A$ \hspace{1.45cm} $H_B$

\vspace{-1.8cm}

\hspace{6.3cm} $G'$

\vspace{5.7cm}

\caption{Extension of our lower bound for $K_4$ detection to $K_{\ell}$ detection, for $\ell \ge 5$. We add a clique $K_{\ell - 4}$ on $\ell -4$ new vertices
to the graph $G'$ and connect every vertex of the clique to every other vertex of $G'$. Then the resulting graph contains a clique on $\ell$ vertices if and only
if the encoded set disjointness instance evaluates to 1, i.e., $x_i = y_i = 1$, for some $i$.}
\label{fig:lb-graphs-K-ell}
\end{figure}

%
The lower bound construction given in Theorem~\ref{thm:lb-two-party} can be extended to the task of detecting $K_{\ell}$, for $\ell \ge 5$ (see also Figure \ref{fig:lb-graphs-K-ell}). To this end, we add a clique on $\ell-4$ new nodes to graph $G'$ (from the proof of Theorem \ref{thm:lb-two-party}) and connect each of these nodes to every vertex in $A \cup B$. Observe that this increases the cut between $A$ and $B$ by $n(\ell-4)$ edges. For $\ell = \Order(\sqrt{n})$, there are only $\Order(n^{3/2})$ additional edges, which implies that the same lower bound as for $K_4$ holds. If $\ell = \omega(\sqrt{n})$, then the number of additional edges is significant, since the size of the cut increases by more than a constant factor. In this case, the round complexity is $\Omega(\frac{n^2}{n(\ell-4) \, \B}) = \Omega(\frac{n}{\ell \, \B})$. Similarly as before, the encoded set disjointness instance evaluates to 1 if and only if $G'$ contains a clique of size $\ell$. We thus conclude with the following theorem.

\begin{theorem}
\label{thm:lb-Kl}
Every algorithm in the \congestb model that detects $K_{\ell}$, for $\ell \ge 4$ and $\ell = \Order(\sqrt{n})$, with probability at least $\frac23$ requires $\Omega(\sqrt{n}/\B)$ rounds. If $\ell = \omega(\sqrt{n})$, then $\Omega(n/(\ell\,\B))$ rounds are required.
\end{theorem}


\section{Lower-bound graph construction}
\label{sec:lb-graph-construction}

In this section, we construct our main technical tool and prove the existence of a $(\Omega(n^2), \Order(n^{3/2}))$-lower-bound graph, see Definition \ref{def:lb-graph}. We will show in Theorem \ref{thm:lower-bound-graph} that Algorithm~1 below constructs a $(\Omega(n^2), \Order(n^{3/2}))$-lower-bound graph with high probability (observe that a non-zero probability already suffices to prove the existence of such a graph).


\subsection{Construction of a $(\Omega(n^2), \Order(n^{3/2}))$-lower-bound graph}
\label{subsec:construction-lower-bound-graphs}
We proceed as follows. We start our construction with a bipartite random graph $G=(A, B, E)$ with $|A| = |B| = n$, where every potential edge $ab$ between $a \in A$ and $b \in B$ is included with probability $p = \frac{1}{\sqrt{n}}$. Observe that for any $a_1, a_2 \in A$ ($a_1 \neq a_2$) and $b_1, b_2 \in B$ ($b_1 \neq b_2$), the probability that $G[\{a_1, a_2, b_1, b_2 \}]$ is isomorphic to a $K_{2,2}$ is $p^4$. We therefore expect $G$ to contain ${n \choose 2}^2 p^4$ copies of $K_{2,2}$, and we prove in Lemma~\ref{lem:bound-1} below that, with high probability, the actual number of copies of $K_{2,2}$ does not deviate significantly from its expectation. Let $\mathcal{K}$ denote the set of copies of $K_{2,2}$ in $G$.

\begin{walgo}
\textbf{Algorithm 1.} Construction of a $(\Omega(n^2), \Order(n^{3/2}))$-lower-bound graph:\\

\textbf{Input:} Integer $n$, let $p = \frac{1}{\sqrt{n}}$.

\begin{enumerate}
\item \textbf{Random Graph:} \\
Let $G = (A, B, E)$ with $|A| = |B| = n$ be the bipartite random graph where
\\\mbox{}\qquad\qquad for every $a \in A, b \in B$ the edge $ab$ is included in $E$ with probability $p$.

Let $\mathcal{K}$ be the family of sets $\{a_1, a_2, b_1, b_2 \}$ with $a_1, a_2 \in A$, $a_1 \ne a_2$, $b_1, b_2 \in B$, $b_1 \ne b_2$
\\\mbox{}\qquad\qquad and $G[\{a_1, a_2, b_1, b_2 \}]$ isomorphic to $K_{2,2}$.

For $S \subseteq A \cup B$, let $\mathcal{K}(S) \subseteq \mathcal{K}$ be the family of subsets $K$ with $S \subseteq K$.

\item \textbf{Peeling Process:} \\
Let $A' \subseteq A$ and $B' \subseteq B$ be a uniform random sample of $A$ and $B$, respectively,
\\\mbox{}\qquad\qquad where every vertex is included with probability $\frac12$.

$\mathcal{H} \gets \{ \}$, $F_A \gets \{ \}$, $F_B \gets \{ \}$.

\textbf{for} every $K = \{a_1, a_2, b_1, b_2 \} \in \mathcal{K}$ \textbf{do}

$\quad$ \textbf{if} $|\mathcal{K}(\{a_1, a_2 \})| \le 6$ and $|\mathcal{K}(\{b_1, b_2 \})| \le 6$ and $|\{a_1, a_2 \} \cap A'| = |\{b_1, b_2 \} \cap B'| = 1$ and
    \\
\mbox{}\qquad\qquad $\{a_1, a_2 \} \notin F_A$ and $\{b_1, b_2 \} \notin F_B$ \textbf{then}

$\quad$ $\quad$ $\mathcal{H} \gets \mathcal{H} \cup K$.

$\quad$ $\quad$ For every $\{a_1, a_2, b_3, b_4 \} \in \mathcal{K}(\{a_1, a_2 \})$, add $\{b_3, b_4 \}$ to $F_B$.

$\quad$ $\quad$ For every $\{a_3, a_4, b_1, b_2 \} \in \mathcal{K}(\{b_1, b_2 \})$, add $\{a_3, a_4 \}$ to $F_A$.

$\quad$ \textbf{end if}

\textbf{end for}

\item \textbf{Lower Bound Graph $H$:} \\
For $K = \{a_1, a_2, b_1, b_2\} \in \mathcal{H}$, let $E_K$ be the edge set $\{a_1b_1, a_1b_2, a_2b_1, a_2b_2 \}$.

\textbf{return} $H := (A, B, \bigcup_{K \in \mathcal{H}} E_K)$.
\end{enumerate}
\end{walgo}

In the peeling phase, we greedily compute a subset $\mathcal{H} \subseteq \mathcal{K}$ such that at the end, the graph induced by the edges of $\mathcal{H}$ is a $(\Omega(n^2), \Order(n^{3/2}))$-lower bound graph. When inserting a set $K = \{a_1, a_2, b_1, b_2 \} \in \mathcal{K}$ into $\mathcal{H}$, we make sure that the following three properties are fulfilled:
\begin{enumerate}
\item We ensure that later on we will never add a $K' = \{a_1', a_2', b_1', b_2' \}$ such that either $\{a_1, a_2, b_1', b_2'\}$ or $\{a_1', a_2', b_1, b_2 \}$ form a $K_{2,2}$. To this end, when inserting $K$ into $\mathcal{H}$, for every $K' \in \mathcal{K}$ that contains the same pair of $A$-vertices (or $B$-vertices), we add its pair of $B$ vertices (resp. pair of $A$ vertices) to set $F_B$ (resp. $F_A$), indicating that this is a forbidden pair. Then, when inserting an element of $\mathcal{K}$ into $\mathcal{H}$, we make sure that its pairs of $A$ and $B$ vertices are not forbidden.

\item We make sure that the insertion of $K$ will not prevent too many other sets $K'$ from being inserted into $\mathcal{H}$. To this end, we guarantee that there are at most six other sets in $\mathcal{K}$ that share the same pair of $A$ vertices and at most six other sets that share the same pair of $B$ vertices. We prove in Lemma~\ref{lem:bound-2} that most $K \in \mathcal{K}$ fulfill this property.

\item It is required that the graphs $G_A$ and $G_B$ as defined in Item~4 of Definition~\ref{def:lb-graph} are bipartite. We therefore partition the sets $A$ and $B$ randomly into subsets $A'$ and $A \setminus A'$, and $B'$ and $B \setminus B'$, and only add $K$ to $\mathcal{H}$ if exactly one of its $A$ vertices is in $A'$ and one of its $B$ vertices is in $B'$.
\end{enumerate}

In the last step of the algorithm, we assemble graph $H$ as the union of the edges contained in the copies of $K_{2,2}$ in $\mathcal{H}$.


\subsection{Analysis of Algorithm 1}

Our analysis relies on some basic properties of the structure of subgraphs of random graphs (for a more complete treatment of related problems, see, e.g., \cite[Chapter~3]{JLR11}). We prove three high probability claims about the construction in Algorithm~1: that the random graph $G$ contains many copies of $K_{2,2}$ (Lemma \ref{lem:bound-1}), that only a small fraction of pairs of $A$ vertices are contained in more than six copies of $K_{2,2}$ (Lemma \ref{lem:bound-2}), and finally that the resulting graph $H$ contains $\Omega(n^2)$ copies of $K_{2,2}$ (Lemma \ref{lem:bound-3}). With these three claims at hand, we will complete the analysis to prove in Theorem~\ref{thm:lower-bound-graph} that with high probability, the output of Algorithm~1 is a $(\Omega(n^2), \Order(n^{3/2}))$-lower-bound graph.

We begin with a proof that in Algorithm~1, the random graph $G$ contains many copies of $K_{2,2}$.

\begin{lemma}
\label{lem:bound-1}
Suppose that $p \ge \frac{1}{n}$. Then there is a constant $C$ such that
\begin{displaymath}
    \Pr \left[ |\mathcal{K}| \le \frac{9}{10} \binom{n}{2}^2 p^4 \right]
        \le
    C \cdot \frac{1}{n^2 p}
    \enspace.
\end{displaymath}
\end{lemma}

\begin{proof}
We will compute the expectation and the variance of $|\mathcal{K}|$ and then use Chebyshev's inequality to bound the probability that $|\mathcal{K}|$ deviates substantially from its expectation.

Let $\mathcal{X}$ be the family of all sets $\{a_1, a_2, b_1, b_2 \}$ with $a_1, a_2 \in A$, $a_1 \ne a_2$, $b_1, b_2 \in B$, $b_1 \ne b_2$, and for $X \in \mathcal{X}$ let $\chi(X)$ be the indicator variable of the event ``$G[X]$ is isomorphic to $K_{2,2}$''. Then:
\begin{displaymath}
    \Exp |\mathcal{K}|
        =
    \sum_{X \in \mathcal{X}} \Pr \left[ \chi(X) = 1 \right]
        =
    |\mathcal{X}| p^4
        =
    \binom{n}{2}^2 p^4
    \enspace,
\end{displaymath}
since $K_{2,2}$ contains $4$ edges. To bound the variance $\Var |\mathcal{K}|$, we use the identity $\Var |\mathcal{K}| = \Exp |\mathcal{K}|^2 - \left( \Exp  |\mathcal{K}| \right)^2$:
\begin{displaymath}
    \Exp |\mathcal{K}|^2
        =
    \Exp \left( \sum_{X \in \mathcal{X}} \chi(X) \right)^2
        =
    \Exp \sum_{X,Y \in \mathcal{X}} \chi(X) \cdot \chi(Y)
        =
    \sum_{X,Y \in \mathcal{X}} \Exp (\chi(X) \cdot \chi(Y))
    \enspace.
\end{displaymath}
We distinguish the following cases:
\begin{itemize}
 \item $|X \cap Y| = 0$. Then, $\Exp (\chi(X) \cdot \chi(Y)) = p^8$. Observe that there are $t_0 = {n \choose 2}^2 {n-2 \choose 2}^2$ such pairs.
 \item $|X \cap Y| = 1$. Then, $\Exp (\chi(X) \cdot \chi(Y)) = p^8$. There are $t_1 = 4 {n \choose 2}^2 {n-2 \choose 2}  {n-2 \choose 1}$ such pairs.
 \item $|X \cap Y| = 2$ and the intersection consists of either two $A$-vertices or two $B$-vertices. Then, $\Exp (\chi(X) \cdot \chi(Y)) = p^8$ and there are $t_{2,1} = 2 \cdot {n \choose 2}^2 {n-2 \choose 2}$ such pairs.
 \item $|X \cap Y| = 2$ and the intersection consists of one $A$-vertex and one $B$-vertex. Then, $\Exp (\chi(X) \cdot \chi(Y)) = p^7$ and there are $t_{2,2} = 4 \cdot {n \choose 2}^2 \cdot (n-2)^2$ such pairs.
 \item $|X \cap Y| = 3$. Then, $\Exp (\chi(X) \cdot \chi(Y)) = p^6$. There are $t_3 = 4 \cdot {n \choose 2}^2 \cdot (n-2)$ such pairs.
 \item $|X \cap Y| = 4$. Then, $\Exp (\chi(X) \cdot \chi(Y)) = p^4$. There are $t_4 = {n \choose 2}^2$ such pairs.
\end{itemize}


\noindent A quick sanity check shows that $t_0 + t_1 + t_{21} + t_{22} + t_3 + t_4 = {n \choose 2}^4$. We thus obtain:
\begin{align*}
    \Var |\mathcal{K}|
       & =
    \Exp |\mathcal{K}|^2 - \left( \Exp |\mathcal{K}| \right)^2
        = p^8 (t_0 + t_1 + t_{2,1}) + p^7 t_{2,2} + p^6 t_3 + p^4 t_4 - {n \choose 2}^4 p^8 \\
       & \le p^7 t_{2,2} + p^6 t_3 + p^4 t_4 = \Order(p^7 n^6) \ ,
\end{align*}
where the last equality holds for every $p \ge \frac{1}{n}$. We apply Chebyshev's inequality and obtain:


\begin{displaymath}
    \Pr \left[\Big||\mathcal{K}| - \Exp|\mathcal{K}|\Big| \ge
            \frac{1}{10} \Exp|\mathcal{K}|\right]
        \le
    \frac{100 \Var |\mathcal{K}|}{ (\Exp |\mathcal{K}|)^2}
        =
    C \cdot \frac{1}{n^2 p}
    \enspace,
\end{displaymath}
for some constant $C$.
\end{proof}

Next, we prove that only a small fraction of pairs of $A$ vertices are contained in more than six copies of $K_{2,2}$.

\begin{lemma}
\label{lem:bound-2}
Let $p = \frac{1}{\sqrt{n}}$. For every constant $\delta > 0$, with high probability, there are at most $(1+\delta) n^2 / 10$ pairs of distinct vertices $a_1, a_2 \in A$ with $|\mathcal{K}(\{a_1, a_2 \})| > 6$.
\end{lemma}

\begin{proof}
Let $a_1, a_2 \in A$, $a_1 \ne a_2$ be arbitrary vertices. Let $B(\{a_1, a_2\}) \subseteq B$ be the set of vertices such that $a_1b, a_2b \in E$. Observe that $|\mathcal{K}(\{a_1, a_2 \})| = \binom{|B(\{a_1, a_2\})|}{2}$. By linearity of expectation, $\Exp |B(\{a_1, a_2\})| = n p^2 = 1$. 

Let $\mathcal{X}$ be the family of all sets of vertices $\{a_1, a_2 \} \subseteq A$ with $a_1 \ne a_2$. Partition now $\mathcal{X}$ into disjoint subsets such that $\mathcal{X} = \mathcal{X}_1 \cup \mathcal{X}_2 \cup \dots \cup \mathcal{X}_{n-1}$, where $|\mathcal{X}_i| = n/2$ and, for every $1 \le i \le n-1$, all elements of $\mathcal{X}_i$ are pairwise disjoint (such a partitioning corresponds to partitioning the complete graph $K_n$ into $n-1$ perfect matchings). For a pair of vertices $P \in \mathcal{X}$, let $\chi(P)$ be the indicator variable of the event ``$|B(P)| \ge 5$''. Recall that $\Exp |B(P)| = n p^2 = 1$ (since $p = 1 / \sqrt{n}$). Hence, by Markov's inequality, we have $\Pr [\chi(P) = 1 ] \le \frac{1}{5}$.

For every $1 \le i \le n-1$ we have $\Exp \sum_{P \in \mathcal{X}_i} \chi(P) \le \frac{1}{5} \frac{n}{2} = \frac{n}{10}$. Observe further that for every $P, Q \in \mathcal{X}_i$, $P \ne Q$, the random variables $B(P)$ and $B(Q)$ are independent. Thus, by a Chernoff bound (for $\mu = \frac{n}{10}$):

\begin{displaymath}
    \Pr\left[|\sum_{S \in \mathcal{X}_i}\chi(S)-\mu| \ge \delta \mu \right]
        \le
    2 \exp \left( - \mu \delta^2 / 3 \right)
        =
    e^{-\Theta(n)}
    \enspace,
\end{displaymath}
for any constant $\delta$. Thus, applying the union bound for every $1 \le i \le n-1$, with high probability, at most $(1+\delta) \frac{n}{10} \cdot (n-1) \le (1+\delta) n^2/10$ pairs of vertices are both connected to at least $5$ vertices of $B$. Hence, at most $(1+\delta) n^2/10$ pairs of vertices $\{a_1, a_2 \}$ are such that $\mathcal{K}(\{a_1, a_2 \}) > {4 \choose 2} = 6$.
\end{proof}

In the next lemma, we show that our resulting graph $H$ contains $\Omega(n^2)$ copies of $K_{2,2}$.

\begin{lemma}
\label{lem:bound-3}
With high probability, the number of copies of $K_{2,2}$ in $H$ is $|\mathcal{H}| = \Omega(n^2)$.
\end{lemma}

\begin{proof}
By Lemma~\ref{lem:bound-1}, we have $|\mathcal{K}| \ge \frac{9}{40}(n-1)^2$ with high probability.
Let $\mathcal{K}' \subseteq \mathcal{K}$ be the subset of sets $\{a_1, a_2, b_1, b_2 \}$ with $\mathcal{K}(\{a_1, a_2 \}) \le 6$ and $\mathcal{K}(\{b_1, b_2 \}) \le 6$. By Lemma~\ref{lem:bound-2}, with high probability, $|\mathcal{K}'| \ge |\mathcal{K}| - 2 \cdot (1+\delta) n^2 / 10$, for any small constant $\delta$.

Let $\mathcal{K}'' \subseteq \mathcal{K}'$ be the subset of sets $\{a_1, a_2, b_1, b_2 \}$ with $|\{a_1, a_2 \} \cap A'| = |\{b_1, b_2 \} \cap A'| = 1$. Observe that every set $X \in \mathcal{K}'$ is included in $\mathcal{K}''$ with probability $\frac{1}{4}$. Thus, by a Chernoff bound, $|\mathcal{K}''| \ge |\mathcal{K}'| / 8$ with high probability.

We argue next that the insertion of any set $K \in \mathcal{K}'$ can block at most $2 \cdot 6^2 = 72$ other sets of $\mathcal{K}'$ from being inserted into $\mathcal{H}$. Consider thus a set $K = \{a_1, a_2, b_1, b_2 \} \in \mathcal{K}'$ that is added to $\mathcal{H}$. This inserts at most six pairs $\{a_3, a_4 \}$ into $F_A$ and six pairs $\{b_3, b_4 \}$ into $F_B$, since $\mathcal{K}(\{a_1, a_2 \}) \le 6$ and $\mathcal{K}(\{b_1, b_2 \}) \le 6$. Since each pair in $F_A$ or in $F_B$ can block at most another six sets of $\mathcal{K}'$, overall at most $2 \cdot 6^2 = 72$ sets of $\mathcal{K}'$ can be blocked by the insertion of $K$ into $\mathcal{H}$.

Hence:
\begin{align*}
    |\mathcal{H}|
        & \ge
    \frac{|\mathcal{K}''|}{72} \ge
    \frac{|\mathcal{K}'|}{8 \cdot 72} \ge
    \frac{(|\mathcal{K}| - 2 \cdot (1+\delta) n^2 / 10)}{8 \cdot 72} \ge
    \frac{(\frac{9}{40}(n-1)^2 - (1+\delta) n^2 / 5)}{8 \cdot 72}
    =
    \Omega(n^2)
    \enspace,
\end{align*}
for $\delta < \frac{1}{8}$.
\end{proof}

With Lemmas \ref{lem:bound-1}--\ref{lem:bound-3} at hand, we are now ready to complete the analysis and show that the graph $H$ fulfills Definition~\ref{def:lb-graph} of a lower bound graph.

\begin{theorem}
\label{thm:lower-bound-graph}
With high probability, the output of Algorithm~1 is a $(\Omega(n^2), \Order(n^{3/2}))$-lower-bound graph.
\end{theorem}

\begin{proof}
We need to check that all items of Definition~\ref{def:lb-graph} are fulfilled with $p = \frac1{\sqrt{n}}$. Concerning Item~1, observe that graph $G$ has $\Order(n^2 p) = \Order(n^{3/2})$ edges with high probability (by a Chernoff bound).

For each $K \in \mathcal{H}$, let $E_K$ denote the edge set added to graph $H$ as in Step~3 of the algorithm. Item~2 holds, since $E(H) = \bigcup_{K \in \mathcal{H}} E_K$, and $H[E_K]$ is isomorphic to $K_{2,2}$, for every $K$, and by Lemma~\ref{lem:bound-3}.

Concerning Item~3, observe that when $K = \{a_1, a_2, b_1, b_2 \}$ is inserted into $\mathcal{H}$, then every $\{a_1, a_2, b_3, b_4 \}$ such that $G[\{a_1, a_2, b_3, b_4 \}]$ is isomorphic to $K_{2,2}$ will not be inserted at a later stage, since $\{b_3, b_4\}$ is inserted into $F_B$. For the same reason, every $\{a_3, a_4, b_1, b_2 \}$ such that $G[\{a_3, a_4, b_1, b_2 \}]$ is isomorphic to $K_{2,2}$ will not be inserted into $\mathcal{H}$. This proves Item~3.

Concerning Item~4, observe that for every $\{a_1, a_2, b_1, b_2 \}$ that is included in $\mathcal{H}$, we have $|\{a_1, a_2\} \cap A'| = |\{b_1, b_2 \} \cap B'| = 1$. Hence, $H_A$ and $H_B$ as defined in Item~4 are bipartite.
\end{proof}


\section{Two-party communication protocol for listing all cliques}
\label{sec:upper-bound}

We consider a two-party communication protocol in the vertex partition model for listing all cliques (of all sizes) in a given graph. The input consists of an undirected graph $G=(V, E)$ with an arbitrary vertex partition $V = V_A \ \dot{\cup} \ V_B$. Let $\cut$ be the $(V_A, V_B)$-cut, $E_A$ be the edge set of $G[V_A]$, and $E_B$ be the edge set of $G[V_B]$. We consider a scenario where Alice is given the subgraph $G_A=(V, E_A \cup \cut) \subseteq G$ and Bob is given $G_B = (V, E_B \cup \cut) \subseteq G$. The objective is for Alice and Bob to detect all cliques (of all sizes) of $G$ and to minimize the number of bits communicated.

We show that in such framework, there is a two-party communication protocol for listing all cliques (of all sizes) that uses $\Order(\sqrt{n} \, \scut)$ bits of communication, where $\cut$ are the edges shared by Alice and Bob. This shows that we cannot improve our lower bounds for the $K_{\ell}$-detection problem, for $\ell = \Order(\sqrt{n})$, in the \congest model (cf. Theorem \ref{thm:lb-Kl}) using the two-party communication framework in the vertex partition model.

Observe that without any communication between the two players, Alice can detect every clique that contains at most one vertex of $V_B$, and, similarly, Bob can detect every clique that contains at most one vertex of $V_A$ (in particular, listing all triangles does not require any communication). Our task is hence to detect every clique consisting of at least two $V_A$ vertices and at least two $V_B$ vertices. We consider two cases:

\begin{enumerate}
\item Suppose that $\scut \ge n^{3/2}$. Then Alice sends all edges $E_A$ to Bob by encoding all entries in the adjacency matrix of $G[V_A]$, which requires at most $n^2 \le \sqrt{n} \scut$ bits. Since Bob then knows the entire graph $G$, he can detect all cliques.

\item Suppose that $\scut < n^{3/2}$. For any vertex $v \in V$, let $d_v$ be the number of edges of $\cut$ incident to $v$, let $V_{\le \sqrt{n}} \subseteq \{ v \in V_A \, : \, d_v \le \sqrt{n} \}$, and let $V_{> \sqrt{n}} = V_A \setminus V_{\le \sqrt{n}}$. We first show how to detect every clique that contains at least one vertex of $V_{\le \sqrt{n}}$. Then, we show how to detect every clique that does not contain any vertex of $V_{\le \sqrt{n}}$.
    \begin{enumerate}
    \item For every $v \in V_{\le \sqrt{n}}$, Bob sends the induced subgraph $G_B[ \Gamma_G(v) \cap V_B]$ (its adjacency matrix) to Alice (observe that Bob knows the set $V_{\le \sqrt{n}}$ without communication). This requires at most $\sqrt{n} \, \scut$ bits, since
        \begin{displaymath}
            \sum_{v \in V_{\le \sqrt{n}}} d_v^2
                \le
            \sqrt{n} \sum_{v \in V_{\le \sqrt{n}}} d_v
                \le
            \sqrt{n} \, \scut
                \enspace.
        \end{displaymath}
        Alice can thus detect any clique that contains at least one vertex of $V_{\le \sqrt{n}}$.

    \item Observe that $|V_{> \sqrt{n}}| \le \frac{\scut}{\sqrt{n}}$. Alice sends the entire subgraph $G_A[V_{> \sqrt{n}}]$ (again, its adjacency matrix) to Bob. This requires at most $\sqrt{n}\,\scut$ bits, since
        \begin{displaymath}
            |V_{> \sqrt{n}}|^2
                \le
            \left( \frac{\scut}{\sqrt{n}} \right)^2
                \le
            \scut \cdot \frac{\scut}{n}
                \le
            \sqrt{n}\scut
                \enspace,
        \end{displaymath}
        using the assumption $\scut \le n^{3/2}$. Bob can thus detect every clique that does not contain any vertex of $V_{\le \sqrt{n}}$.
    \end{enumerate}
\end{enumerate}

We thus obtain the following theorem:

\begin{theorem}
\label{thm:upper-bound}
There is a two-party communication protocol in the vertex partition model for listing all cliques (of all sizes) that uses $\Order(\sqrt{n}\,\scut)$ communication rounds, where $\cut$ is the set of shared edges between Alice and Bob.
\end{theorem}


\section{Conclusions}
\label{sec:conclusions}

In this paper, we gave a non-trivial lower bound for the problem of detecting a clique $K_{\ell}$, for $\ell \ge 4$, in the classical distributed \congest model. We show that detecting $K_{\ell}$ requires $\Omega(\frac{n}{(\ell + \sqrt{n}) \, \B})$ communication rounds, for every $\ell \ge 4$, where $\B$ is the bandwidth of the communication links. Our lower bound is complemented by a matching upper bound obtained by a two-party communication protocol in the vertex partition model for listing all cliques (of all sizes). This demonstrates that our lower bound cannot be improved using the two-party communication framework.

We leave as a great open question whether the complexity of clique detection in the \congest model is sublinear, or one needs $\ThetaT(n)$ communication rounds to detect even a copy of $K_4$. Since it seems that the two-party communication approach used in our lower bound cannot be improved further, we do not have any intuition whether the lower bound is tight, or could be improved significantly. On the other hand, the recent $\OrderT(n^{2/3})$-communication rounds algorithm for detecting a triangle \cite{IL17} raises some hopes that maybe also $K_4$ could be detected in a sublinear number of rounds.





\newcommand{\Proc}{Proceedings of the\xspace}
\newcommand{\STOC}{Annual ACM Symposium on Theory of Computing (STOC)}
\newcommand{\FOCS}{IEEE Symposium on Foundations of Computer Science (FOCS)}
\newcommand{\SODA}{Annual ACM-SIAM Symposium on Discrete Algorithms (SODA)}
\newcommand{\COCOON}{Annual International Computing Combinatorics Conference (COCOON)}
\newcommand{\DISC}{International Symposium on Distributed Computing (DISC)}
\newcommand{\ESA}{Annual European Symposium on Algorithms (ESA)}
\newcommand{\ICALP}{Annual International Colloquium on Automata, Languages and Programming (ICALP)}
\newcommand{\IPL}{Information Processing Letters}
\newcommand{\JACM}{Journal of the ACM}
\newcommand{\JALGORITHMS}{Journal of Algorithms}
\newcommand{\JCSS}{Journal of Computer and System Sciences}
\newcommand{\OPODIS}{International Conference on Principles of Distributed Systems (OPODIS)}
\newcommand{\PODC}{Annual ACM Symposium on Principles of Distributed Computing (PODC)}
\newcommand{\SICOMP}{SIAM Journal on Computing}
\newcommand{\SPAA}{Annual ACM Symposium on Parallelism in Algorithms and Architectures (SPAA)}
\newcommand{\STACS}{Annual Symposium on Theoretical Aspects of Computer Science (STACS)}
\newcommand{\TALG}{ACM Transactions on Algorithms}
\newcommand{\TCS}{Theoretical Computer Science}

\bibliography{references}

\begin{thebibliography}{10}

\bibitem{BP11}
Zvika Brakerski and Boaz Patt{-}Shamir.
\newblock Distributed discovery of large near-cliques.
\newblock {\em Distributed Computing}, 24(2):79--89, 2011.

\bibitem{CFSV16}
Keren Censor{-}Hillel, Eldar Fischer, Gregory Schwartzman, and Yadu Vasudev.
\newblock Fast distributed algorithms for testing graph properties.
\newblock In {\em \Proc 30th \DISC}, pages 43--56, 2016.

\bibitem{CKKLPS15}
Keren Censor{-}Hillel, Petteri Kaski, Janne~H. Korhonen, Christoph Lenzen, Ami
  Paz, and Jukka Suomela.
\newblock Algebraic methods in the congested clique.
\newblock In {\em \Proc 35th \PODC}, pages 143--152, 2015.

\bibitem{CKP17}
Keren Censor{-}Hillel, Seri Khoury, and Ami Paz.
\newblock Quadratic and near-quadratic lower bounds for the {CONGEST} model.
\newblock In {\em \Proc 31st \DISC}, pages 10:1--10:16, 2017.

\bibitem{DLP12}
Danny Dolev, Christoph Lenzen, and Shir Peled.
\newblock ``{T}ri, tri again'': Finding triangles and small subgraphs in a
  distributed setting.
\newblock In {\em \Proc 26th \DISC}, pages 195--209, 2012.

\bibitem{DKO14}
Andrew Drucker, Fabian Kuhn, and Rotem Oshman.
\newblock On the power of the congested clique model.
\newblock In {\em \Proc 33rd \PODC}, pages 367--376, 2014.

\bibitem{EFFGLMMOORT17}
Guy Even, Orr Fischer, Pierre Fraigniaud, Tzlil Gonen, Reut Levi, Moti Medina,
  Pedro Montealegre, Dennis Olivetti, Rotem Oshman, Ivan Rapaport, and Ioan
  Todinca.
\newblock Three notes on distributed property testing.
\newblock In {\em \Proc 31st \DISC}, pages 15:1--15:30, 2017.

\bibitem{FGO17b}
Orr Fischer, Tzlil Gonen, and Rotem Oshman.
\newblock Superlinear lower bounds for distributed subgraph detection.
\newblock {\em CoRR}, abs/1711.06920, 2017.

\bibitem{FO17}
Pierre Fraigniaud and Dennis Olivetti.
\newblock Distributed detection of cycles.
\newblock In {\em \Proc 29th \SPAA}, pages 153--162, 2017.

\bibitem{GO17}
Tzlil Gonen and Rotem Oshman.
\newblock Lower bounds for subgraph detection in the {CONGEST} model.
\newblock In {\em \Proc 21st \OPODIS}, pages 6:1--6:16, 2017.

\bibitem{IL17}
Taisuke Izumi and Fran{\c{c}}ois~Le Gall.
\newblock Triangle finding and listing in {CONGEST} networks.
\newblock In {\em \Proc 37th \PODC}, pages 381--389, 2017.

\bibitem{JLR11}
Svante Janson, Tomasz {\L}uczak, and Andrzej Ruci\'{n}ski.
\newblock {\em Random Graphs}.
\newblock John Wiley \& Sons, 2011.

\bibitem{KS92}
Bala Kalyanasundaram and Georg Schnitger.
\newblock The probabilistic communication complexity of set intersection.
\newblock {\em SIAM Journal on Discrete Mathematics}, 5(4):545--557, 1992.

\bibitem{KR17}
Janne~H. Korhonen and Joel Rybicki.
\newblock Deterministic subgraph detection in broadcast {CONGEST}.
\newblock In {\em \Proc 21st \OPODIS}, pages 4:1--4:16, 2017.

\bibitem{KN97}
Eyal Kushilevitz and Noam Nisan.
\newblock {\em Communication Complexity}.
\newblock Cambridge University Press, 1997.

\bibitem{POS16}
Gopal Pandurangan, Peter Robinson, and Michele Scquizzato.
\newblock Tight bounds for distributed graph computations.
\newblock {\em CoRR}, abs/1602.08481, 2016.

\bibitem{P00}
David Peleg.
\newblock {\em Distributed Computing: A Locality-Sensitive Approach}.
\newblock SIAM Monographs on Discrete Mathematics and Applications. SIAM,
  Philadelphia, PA, 2000.

\end{thebibliography}


\end{document}